\documentclass{article}
\usepackage{graphicx}
\usepackage{float}
\usepackage{amsmath, amsthm, latexsym}
\usepackage{amssymb}

\newtheorem{theorem}{Theorem}
\newtheorem{lemma}{Lemma}

\newtheorem{corollary}{Corollary}

\newtheorem{remark}{Remark}
\newtheorem{example}{Example}
\usepackage{pdfsync}
\usepackage{hyperref}      
\begin{document}

\vspace*{3cm} \thispagestyle{empty}
\vspace{5mm}

\noindent \textbf{\Large Velocity addition formulas in Robertson-Walker spacetimes}\\

\textbf{\normalsize  \textbf{\normalsize David Klein and Jake Reschke}\footnote{Department of Mathematics and Interdisciplinary Research Institute for the Sciences, California State University, Northridge, Northridge, CA 91330-8313. Email: david.klein@csun.edu, jake.reschke.244@my.csun.edu}}
\\

\vspace{4mm} \parbox{11cm}{\noindent{\small Universal velocity addition formulas analogous to the well-known formula in special relativity are found for four geometrically defined relative velocities in a large class of Robertson-Walker spacetimes. Explicit examples are given.  The special relativity result is recovered as a special case, and it is shown that the spectroscopic relative velocity, in contrast to three other geometric relative velocities, follows the same addition law as in special relativity for comoving observers in Robertson-Walker cosmologies. }\\


}\\
\vspace{6cm}
\pagebreak

\setlength{\textwidth}{27pc}
\setlength{\textheight}{43pc}

\section{Introduction}
In special relativity, the velocity addition formula for three inertial observers whose relative motions are spatially collinear may be expressed as,

\begin{equation}\label{special}
v_{3}=\frac{v_{1}+v_{2}}{1+v_{1}v_{2}}.
\end{equation}
Here and below, relative to a central observer, $v_{1}$ is the velocity of a secondary observer; $v_{2}$ is the velocity of a test particle relative to this secondary observer; and $v_{3}$ is the velocity of the test particle relative to the central observer.  An alternative and sometimes useful equivalent relationship is,

\begin{equation}\label{special'}
\frac{1-v_{3}}{1+v_{3}}=\left(\frac{1-v_{1}}{1+v_{1}}\right)\left(\frac{1-v_{2}}{1+v_{2}}\right).
\end{equation}

\noindent Are meaningful velocity addition formulas, analogous to these, possible in general relativity?  At first glance, the answer appears to be no. General relativity provides no \textit{a priori} definition of relative velocity when the test particles and observers are located at different space-time points.  Different coordinate charts give rise to different conceptions of relative velocity.\\

\noindent   To avoid such ambiguities, we study the coordinate independent, purely geometrically defined Fermi, kinematic, spectroscopic, and astrometric, relative velocities introduced by V. J. Bol\'os in \cite{bolos}, and subsequently developed in  \cite{KC10}, \cite{randles}, \cite{Bolos12}, \cite{sam}, \cite{klein13}, \cite{Bolos12b}, \cite{Bolos13}, \cite{Bolos14}. The definitions of these geometric relative velocities depend on two distinct notions of simultaneity: ``spacelike simultaneity''  (or ``Fermi simultaneity'') and ``lightlike simultaneity.''\\  

\noindent The Fermi and kinematic relative velocities depend on spacelike simultaneity.  Two events are simultaneous in this sense if they lie on the same Fermi space slice $\mathcal{M}_{\tau}$ determined by a fixed Fermi time coordinate $\tau$.  To define Fermi coordinates (see, e.g., \cite{KC1}), consider a foliation of some neighborhood $\mathcal{U}$ (which might be the entire spacetime) of a central observer's geodesic worldline, $\beta_{0}(t)$, by disjoint Fermi spaceslices $\{\mathcal{M}_{\tau}\}$ defined by,

\begin{equation}\label{slice2}
\mathcal{M}_{\tau}\equiv \varphi_{\tau}^{-1}(0).
\end{equation}
Here the function $\varphi_{\tau}: \mathcal{U} \rightarrow \mathbb{R}$ is given by,
\begin{equation}\label{slice}
\varphi_{\tau}(p)=g(\exp_{\beta(\tau)}^{-1}p,\, \dot{\beta}_{0}(\tau)),
\end{equation}
where the overdot represents differentiation with respect to proper time $\tau$ along $\beta_{0}$, $g$ is the metric tensor, and the exponential map, $\exp_{p}(v)$ denotes the evaluation at affine parameter $1$ of the geodesic starting at  point $p\in\mathcal{M}$, with initial derivative $v$. In other words, the Fermi spaceslice $\mathcal{M}_{\tau}$ of all $\tau$-simultaneous points consists of all the spacelike geodesics orthogonal to the path of the central observer $\beta_{0}$ at fixed proper time $\tau$. A convenient coordinate system on $\mathcal{M}_{\tau}$ consists of two angular coordinates together with the proper distance $\rho$ from the central observer's path to a point on $\mathcal{M}_{\tau}$ (see \cite{randles, klein13}). \\ 

\noindent The Fermi relative velocity, $V_{\mathrm{Fermi}}$, of a test particle relative to an observer, $\beta_{0}(\tau)$, is a vector field along $\beta_{0}(\tau)$, orthogonal to the 4-velocity $U$ of the observer (and therefore tangent to $\mathcal{M}_{\tau}$) at each proper time $\tau$ of the observer.  For a test particle undergoing purely radial motion (which is our focus), the magnitude of $V_{\mathrm{Fermi}}$, or the Fermi relative speed, $v_{\mathrm{Fermi}}$, is the rate of change of proper distance $\rho$ of the test particle away from the  observer $\beta(\tau)$, with respect to proper time $\tau$.   Eq. \eqref{special} is thus a statement about Fermi relative velocities in the context of special relativity.  For more general motion (i.e. non radial), the definition of Fermi relative velocity may be found in \cite{bolos}.\\

\noindent The kinematic velocity of a test particle at a spacetime point $q$, relative to the central observer $\beta_{0}(\tau)$, is found by first parallel transporting the test particle's 4-velocity $u'$ along a radial spacelike geodesic (lying on a Fermi space slice) to the 4-velocity denoted by $\tau_{q p}u'$ in the tangent space of the central observer at spacetime point $p=\beta_{0}(\tau_{0})$. If the 4-velocity of $\beta_{0}$ at the point $p$ is $u$, then the kinematic relative velocity $v_{\mathrm{kin}}$ of the test particle is defined to be the unique vector orthogonal to $u$, in the tangent space at $p$, satisfying $\tau_{q p}u'=\gamma (u+v_{\mathrm{kin}})$, where the scalar $\gamma$ is uniquely determined and is given by $\gamma=1/\sqrt{1-v_{\mathrm{kin}}^{2}}$.\\

\noindent The spectroscopic (or barycentric) and astrometric relative velocities can be found, in principle, from spectroscopic and astronomical observations. These two relative velocities depend on lightlike simultaneity, according to which two events are simultaneous if they both lie in the past-pointing horismos at the spacetime point $p$ of the central observer (the past-pointing horismos is tangent to the backward light cone).  More specifically, using the notation above, let  $\psi_{\textsc{t}} : \mathcal{U} \rightarrow \mathbb{R}$ be given by,
\begin{equation}\label{slice}
\psi_{\textsc{t}}(p)=g(\exp_{\beta_{0}(\textsc{t})}^{-1}p,\, \exp_{\beta_{0}(\textsc{t})}^{-1}p),
\end{equation}
where in this context, we denote proper time along $\beta_{0}$ by $\textsc{t}$.\footnote{In Sect.\ref{opticalcoords} the symbol $\textsc{t}$ will denote the time coordinate in optical coordinates which is different from the Fermi time coordinate $\tau$. However, these two coordinate times are identical along the path $\beta_{0}$ where they are both proper time.}  Now define
\begin{equation}
E_{\textsc{t}}\equiv\psi_{\textsc{t}}^{-1}(0)-\{\beta_{0}(\textsc{t})\}.
\end{equation}
The 3-dimensional submanifold, $E_{\textsc{t}}$, is called the horismos submanifold at the spacetime point $\beta_{0}(\textsc{t})$ (where proper time $\textsc{t}$ is fixed).  An event $q$ is in  $E_{\tau}$ if and only if $q \neq\beta_{0}(\textsc{t})$ and there exists a lightlike geodesic joining $\beta_{0}(\textsc{t})$ and $q$.  $E_{\textsc{t}}$ has two connected components,
$E_{\textsc{t}}^{+}$ and $E_{\textsc{t}}^{-}$, respectively, the future-pointing and past-pointing horismos submanifolds of $\beta_{0}(\textsc{t})$ (see \cite{Bolos05} and the references  therein).\\

\noindent We note that for a comoving observer $\beta_{0}$, it was proved in \cite{klein13} (see also \cite{carney}) that for a broad class of Robertson-Walker cosmologies, the maximal chart for Fermi coordinates is exactly the causal past of $\beta_{0}$ so that,

\begin{equation}\label{causalpast}
\bigcup_{\tau>0}\mathcal{M}_{\tau}= \bigcup_{\textsc{t}>0}E^-_{\textsc{t}}\cup\beta_{0}.
\end{equation}

\noindent The spectroscopic relative velocity $v_{\mathrm{spec}}$ is calculated analogously to $v_{\mathrm{kin}}$, described above, except that the 4-velocity $u'$ of the test particle is parallel transported to the tangent space of the observer along a null geodesic lying on the past-pointing horismos of the central observer, instead of along the Fermi space slice.  The astrometric relative velocity, $v_{\mathrm{ast}}$, of a test particle whose motion is purely radial is calculated analogously to $v_{\mathrm{Fermi}}$, as the rate of change of the affine distance (see Sect. \ref{opticalcoords}), which corresponds to the \textit{observed} proper distance (through light signals at the time of observation) with respect to the proper time of the observer, as may be done via parallax measurements. A complete description is given in \cite{bolos}.\\

\noindent In Minkowski spacetime, the coordinates of three Lorentz frames (with relative velocities as in Eq. \eqref{special}) in ``standard configuration'' relative to each other are related by Lorentz boosts along a given space axis of a central observer, and all origins of coordinates coincide when all of the time coordinates equal zero.  The natural generalization to Robertson-Walker cosmologies is to consider three observers co-moving with the Hubble flow, with collinear relative motion (i.e. with two fixed space coordinates). In the case of the Milne universe, this configuration amounts to three special relativistic observers in standard configuration.\\  

\noindent In this paper, we obtain generalizations of Eq \eqref{special} for comoving observers in Robertson-Walker spacetimes.  As a special case, using general methods we recover Eq \eqref{special} for the Milne universe which is diffeomorphic to the forward light cone of Minkowski spacetime under a simple change of coordinates. In the Milne universe, Eq. \eqref{special} holds for all of the geometrically defined relative velocities except for the astrometric velocity, which follows a different addition law (Eq. \eqref{milneastroadd}).   The addition laws for the four geometrically defined relative velocities are strikingly different. For example, we show that for lightlike simultaneous observers, the spectroscopic relative velocity follows Eq \eqref{special} in all Robertson-Walker spacetimes considered in this paper, while the addition formulas for the other relative velocities are quite different.\\

\noindent Unlike Minkowski spacetime, for a general Robertson-Walker cosmology, the distinct foliations by sets of simultaneous events for each of the three observers forces the existence of more than one velocity addition formula, depending on how comparisons are made. \\

\noindent This paper is organized as follows.  Sections \ref{RWmetric} and \ref{3observers} develop notation for Fermi and optical coordinates and spacetime paths of observers and test particles.  In Sections \ref{Fermicoords} and \ref{opticalcoords} we express the Robertson-Walker metric in Fermi coordinates and optical coordinates and slightly generalize formulas for relative velocities found in \cite{Bolos12} that we use in the proofs of our theorems.  Sections \ref{spacelikesection} and \ref{lightadd} give general velocity addition formulas for the Fermi, kinematic, spectroscopic, and astrometric relative velocities. In Section \ref{hubble} we show how, in special cases, one may use relationships between the Hubble velocity and the four geometrically defined relative velocities to derive velocity addition formulas.  We use that to derive velocity addition formulas for Robertson-Walker cosmologies with any power law scale factor. Examples for the de Sitter universe and particular power law cosmologies are given in Section \ref{examples}, and Section \ref{conclude} gives concluding remarks.

\section{The Robertson-Walker metric}\label{RWmetric}

\noindent The Robertson-Walker metric on space-time $\mathcal{M}$ is given by the line element,

\begin{equation}\label{frwmetric}
ds^2=-dt^2+a^2(t)\left[d\chi^2+S^2_k(\chi)d\Omega^2\right],
\end{equation}

\noindent where $d\Omega^2=d\theta^{2}+\sin^{2}\theta \,d\varphi^{2}$, $a(t)$ is the scale factor, and, 

\begin{equation}\label{Sk}
S_{k}(\chi)=
\begin{cases}
\sin\chi &\text{if}\,\, k= 1\\
 \chi&\text{if}\,\,k= 0\\ 
 \sinh\chi&\text{if}\,\,k=-1.
\end{cases}
\end{equation}

\noindent The coordinate $t>0$ is cosmological time and $\chi, \theta, \varphi$ are dimensionless. The values $+1,0,-1$ of the parameter $k$ distinguish the three possible maximally symmetric space slices for constant values of $t$ with positive, zero, and negative curvatures respectively.  The radial coordinate $\chi$ takes all positive values for $k=0$ or $-1$, but is bounded above by $\pi$ for $k=+1$.\\

\noindent We assume throughout that $k= 0$ or $-1$ so that the range of $\chi$ is unrestricted. The techniques employed for these two cases may be extended to the case $k=+1$ with the additional restriction that $\chi<\pi$ so that spacelike geodesics do not intersect.\\

\noindent There is a coordinate singularity in \eqref{frwmetric} at $\chi =0$, but this will not affect the calculations that follow. Since our intention is to study radial motion with respect to a central observer, it suffices to consider the $2$-dimensional Robertson-Walker metric given by
\begin{equation}
\label{metric}
\mathrm{d}s^2=-\mathrm{d}t^2+a^2(t)\mathrm{d}\chi ^2,
\end{equation}
for which there is no singularity at $\chi=0$ and we may allow $\chi\in(-\infty, \infty)$.  We assume throughout that $a(t)$ is a smooth, increasing function of $t>0$.

\section{Notation for three observers}\label{3observers}

We denote a central observer by the path $\beta_0(t_0)=(t_0,0)$, a secondary observer by $\beta_1(t_1)=(t_1,\chi_1)$, and a test particle by $\beta_2(t_2)=(t_2,\chi_2)$. Each spacetime path is parameterized by its proper time, and we take $\beta_1$ and $\beta_2$ to be comoving so that $\chi_1$ and $\chi_2$ are constant. Without loss of generality we take $\chi_2$ to be positive, while allowing $\chi_1$ to assume both positive and negative values. We allow for the possibility that $\chi_1>\chi_2$.\\

\noindent The Fermi, kinematic, spectroscopic and astrometric velocities of a test particle relative to a particular observer, $\beta(\tau)$, are smooth vector fields defined along the path $\beta(\tau)$.  Those vector fields (along paths of observers) are denoted by upper case letters, and scalar fields (relative speeds) by lower case letters.\\

\noindent In what follows, the relative velocity vector fields are defined on the spacetime paths $\beta_0$ and $\beta_1$, i.e. the central and secondary observers. From their definitions (see \cite{bolos}) all four relative velocity vector fields are spacelike and orthogonal to the 4-velocities of the observers.  They are each multiples of the unit vector field 
\begin{equation}
\mathcal{S}=\frac{1}{a(t)}\frac{\partial}{\partial\chi}. 
\end{equation}
The Fermi, kinematic, spectroscopic and astrometric velocity vector fields are denoted respectively by 
$V_{\mathrm{Fermi}}=v_{\mathrm{Fermi}}\mathcal{S}$, $V_{\mathrm{kin}}=v_{\mathrm{kin}}\mathcal{S}$,$V_{\mathrm{spec}}=v_{\mathrm{spec}}\mathcal{S}$, and $V_{\mathrm{ast}}=v_{\mathrm{ast}}\mathcal{S}$.  We note that the scalars $v_{\mathrm{Fermi}}$, $v_{\mathrm{kin}}$, etc. can be positive or negative, indicating velocities in the same or opposition direction as $\mathcal{S}$. We note that this notation differs from that used in \cite{Bolos12}.\\

\noindent Subscripts will be used to distinguish the different relative velocities. Subscripts 1, 2 and 3 will denote a velocity of $\beta_1$ relative to $\beta_0$, $\beta_2$ relative to $\beta_1$ and $\beta_2$ relative to $\beta_0$, respectively. E.g. $v_{\mathrm{Fermi}2}$ is the Fermi velocity of the test particle relative to the secondary observer.

\section{Fermi coordinates and spacelike simultaneity}\label{Fermicoords}

\noindent It may be shown that the metric \eqref{metric} expressed in Fermi coordinates $(\tau, \rho)$ for the central observer, $\beta_{0}$, (whose $\chi$-coordinate is taken to be zero) has the form,
\begin{equation}\label{fermipolar}
\mathrm{d}s^2=g_{\tau\tau}(\tau, \rho) \mathrm{d}\tau^2+\mathrm{d}\rho^2.
\end{equation}
General formulas for $g_{\tau \tau }$ were derived in \cite{randles,Bolos12,klein13}.  Non negative values of the spatial coordinate $\rho$ give the proper distance along spacelike geodesics orthogonal to the world line of the comoving Fermi or central observer.  However, consistent with Eq. \eqref{metric} we also allow $\rho$ to take negative values. \\

\noindent Formulas for relative velocities developed in \cite{randles,Bolos12} can be easily extended to accommodate negative values of the spatial coordinates.  At a given proper time $\tau$ of the central observer, suppose that the Fermi coordinates of a test particle are $(\tau, \rho)$ (and the coordinates of the central observer are $\beta_{0}(\tau)=(\tau,0)$).  Let the curvature coordinates of the test particle be $(t,\chi)$ (where $\chi$ is fixed). Then the kinematic velocity of a comoving test particle relative to the central observer is given by,

\begin{equation}
\label{eqvkin}
 v_{\mathrm{kin}} =\mathrm{sgn}(\chi)\sqrt{1-\frac{a^2(t)}{a^2(\tau )}}.
\end{equation}
\noindent The Fermi speed is given by $v_{\mathrm{Fermi}}= d\rho/d\tau $ for a radially moving test particle.\\

\noindent The common dependence on spacelike simultaneity of the Fermi and kinematic relative velocities allows for a direct comparison of these two notions of relative velocity of a test particle at a given spacetime point.  The Fermi and kinematic speeds of a radially moving test particle at the spacetime point $(\tau, \rho)$ in the Fermi coordinates of a central observer are related by,

\begin{equation}\label{vmetric}
v_{\mathrm{Fermi}}=\sqrt{-g_{\tau\tau}(\tau,\rho)}\,v_{\mathrm{kin}}.
\end{equation}

\noindent In what follows we will need Fermi coordinates for a comoving observer with some fixed $\chi$-coordinate, $\chi_{j}$, not necessarily equal to zero.  As in the previous section, let $\beta_j$ represent the worldline of a comoving observer whose $\chi$ coordinate is $\chi_j$. The proper time of $\beta_j$ is again $\tau_{j}$, and Fermi coordinates for the observer $\beta_j$ may be constructed.  In this case we denote the leading metric coefficient by $g_{\tau_{j}\tau_{j}}$, so that, for example, $g_{\tau\tau}$ may also be written as $g_{\tau_{0}\tau_{0}}$. \\

\noindent Functional relationships among the coordinates, $\tau_j$, $\chi_j$, $t$, and $\chi$, follow from the observation that the vector field,

\begin{equation}\label{spacegeo}
X=-\mathrm{sgn}(\chi-\chi_j)\sqrt{\left(\frac{a(\tau_j)}{a(t)}\right)^2-1}\frac{\partial}{\partial t}+\frac{a(\tau_j)}{a^2(t)}\frac{\partial}{\partial \chi},
\end{equation}

\noindent is geodesic, unit, spacelike and orthogonal to the  4-velocity of $\beta_j$, i.e. $X$ is tangent to $\mathcal{M}_{\tau_j}$ for each value of $\tau_{j}$. Let $q=(t,\chi)\in\mathcal{M}_{\tau_j}$. Then there exists an integral curve (for the vector field $X$ or $-X$) from $p\in\beta_{j}$ to $q$, and thus using Eq. \eqref{spacegeo} we can find a relationship between $\tau_j$, $\chi_j$, $t$, and $\chi$ (see \cite{Bolos12}):

\begin{equation}
\int_t^{\tau_j}\frac{a(\tau_j)}{a(\tilde{t})}\frac{1}{\sqrt{a^2(\tau_j)-a^2(\tilde{t})}}\mathrm{d}\tilde{t}=|\chi-\chi_j|.\\
\end{equation}

\begin{remark}\label{spaceobservers}  Eq.\eqref{eqvkin} can be easily adapted to yield velocities of a test particle relative to the comoving observer, $\beta_{j}$  by replacing $\mathrm{sgn}(\chi)$ by  $\mathrm{sgn}(\chi-\chi_{j})$ and $\tau$ by $\tau_{j}$.
\end{remark}

\section{Optical coordinates and lightlike simultaneity}\label{opticalcoords}

In the framework of lightlike simultaneity, it will be convenient to use optical coordinates (also known as observational coordinates).  Following \cite{Bolos12}, we may express the Robertson-Walker metric in the optical coordinates, $(\textsc{t}, \delta)$, determined by the central observer $\beta_{0}$,
\begin{equation}
\label{opticalmetric}
\mathrm{d}s^2=\tilde{g}_{\textsc{t}\textsc{t}}\mathrm{d}\textsc{t} ^2+2\,\mathrm{sgn}(\delta)\,\mathrm{d}\textsc{t}\, \mathrm{d}\delta\equiv -2\left( 1-\frac{\dot{a}(\textsc{t})}{a(\textsc{t})}|\delta| -\frac{1}{2}\frac{a^2\left( t\right)}{a^2(\textsc{t} )}\right)\mathrm{d}\textsc{t} ^2+2\,\mathrm{sgn}(\delta)\,\mathrm{d}\textsc{t}\, \mathrm{d}\delta ,
\end{equation}
where $t(\textsc{t} ,\delta )$ is given implicitly by,

\begin{equation}
\label{delta}
\delta=\delta (t,\chi)=\mathrm{sgn}(\chi)\int _{t}^{\textsc{t} (t,\chi)} \frac{a(u)}{a\left( \textsc{t} (t,\chi)\right) }\,\mathrm{d}u.
\end{equation}
and $\textsc{t}(t,\chi)$ is determined implicitly from the equation,

\begin{equation}
\label{eqt1b}
\int _{t}^{\textsc{t} } \frac{1}{a(u)}\,\mathrm{d}u=|\chi|.
\end{equation}

\begin{remark}\label{times}
Along the path, $\beta_{0}$, of a comoving central observer, the three time coordinates $t, \tau, \textsc{t}$ (respectively from curvature coordinates, Fermi coordinates, and optical coordinates) take the same values and are identical. They differ, however, at spacetime points off the path $\beta_{0}$. We note that our notation for the optical time coordinate $\textsc{t}$ differs from that used in \cite{Bolos12}.  We have also dropped the subscript $\ell$ in $t$ and $\chi$ used in that reference.
\end{remark}

\noindent Formulas for relative velocities developed in \cite{Bolos12} can be readily extended to accommodate negative values of the spatial coordinates. The spectroscopic and astrometric velocities of a comoving test particle at fixed spatial coordinate $\chi$ relative to the central observer, $\beta_{0}(\textsc{t})$, are given by,

\begin{equation}
\label{eqvspec}
v_{\mathrm{spec}}=\mathrm{sgn}(\chi)\frac{a^2(\textsc{t} )-a^2(t)}{a^2(\textsc{t} )+a^2(t)}
\end{equation}
and
\begin{equation}
\label{eqvast}
 v_{\mathrm{ast}}=\mathrm{sgn}(\chi)\left(1-|\delta| \frac{\dot{a}(\textsc{t} )}{a(\textsc{t} )}-\frac{a^2(t)}{a^2(\textsc{t} )}\right).
\end{equation}
where $\delta $ is the affine distance parameter from the central observer to the comoving test particle.  As in the previous section, the formulas here for relative velocities  can be easily adapted to yield velocities of a test particle relative to the comoving observer, $\beta_{j}$ by replacing $\mathrm{sgn}(\chi)$ by  $\mathrm{sgn}(\chi-\chi_{j})$ and $\tau$ by $\tau_{j}$.\\

\noindent The following comparison of the astrometric and spectroscopic relative velocities is also possible (see \cite{Bolos12}) because of their common dependence on timelike simultaneity,  

\begin{equation}
\label{vastend2}
v_{\mathrm{ast}}=-\frac{\mathrm{sgn}(\delta)}{2}\left(\tilde{g}_{\textsc{t}\textsc{t}}(\textsc{t},\delta)+
\frac{1-\mathrm{sgn}(\delta)\,v_{\mathrm{spec}}}{1+\mathrm{sgn}(\delta)\,v_{\mathrm{spec}}}\right) .
\end{equation}

\section{Velocity addition formulas for spacelike simultaneity}\label{spacelikesection}

\noindent In this section we derive velocity addition formulas for the Fermi and kinematic relative velocities of comoving observers. In general there can be no single formula like Eq.\eqref{special} for either relative velocity.  This is because the respective Fermi submanifolds of simultaneous events for the central and secondary observers do not coincide, and as a consequence the spacetime points at which relative velocities are calculated can be reasonably chosen in more than one way.  We show, however, for the special case of the Milne universe (essentially for the case of special relativity), that Eq.\eqref{special} follows as a special case of Theorem \ref{spacelike3} and Corollary \ref{Fermiadd} below. \\  

\noindent The figures below depict three general scenarios for which addition formulas may be deduced, and we refer to them in Theorem \ref{spacelike3} and Corollary \ref{Fermiadd}.  In each figure, the comoving paths $\beta_{i}$ as described in Sect. \ref{3observers} are indicated.  The vertical axis is cosmological time, which is also proper time for each comoving path, and the horizontal axis is the $\chi$-axis.  The curves labelled by $\Psi_{0}$ or $\Psi_{0}^{\pm}$ are spacelike geodesics orthogonal to the central observer $\beta_{0}$; $\Psi_{1}$ is a spacelike geodesic orthogonal to the secondary observer $\beta_{1}$.  

\begin{figure}[H]
\centering
\includegraphics[scale=0.3]{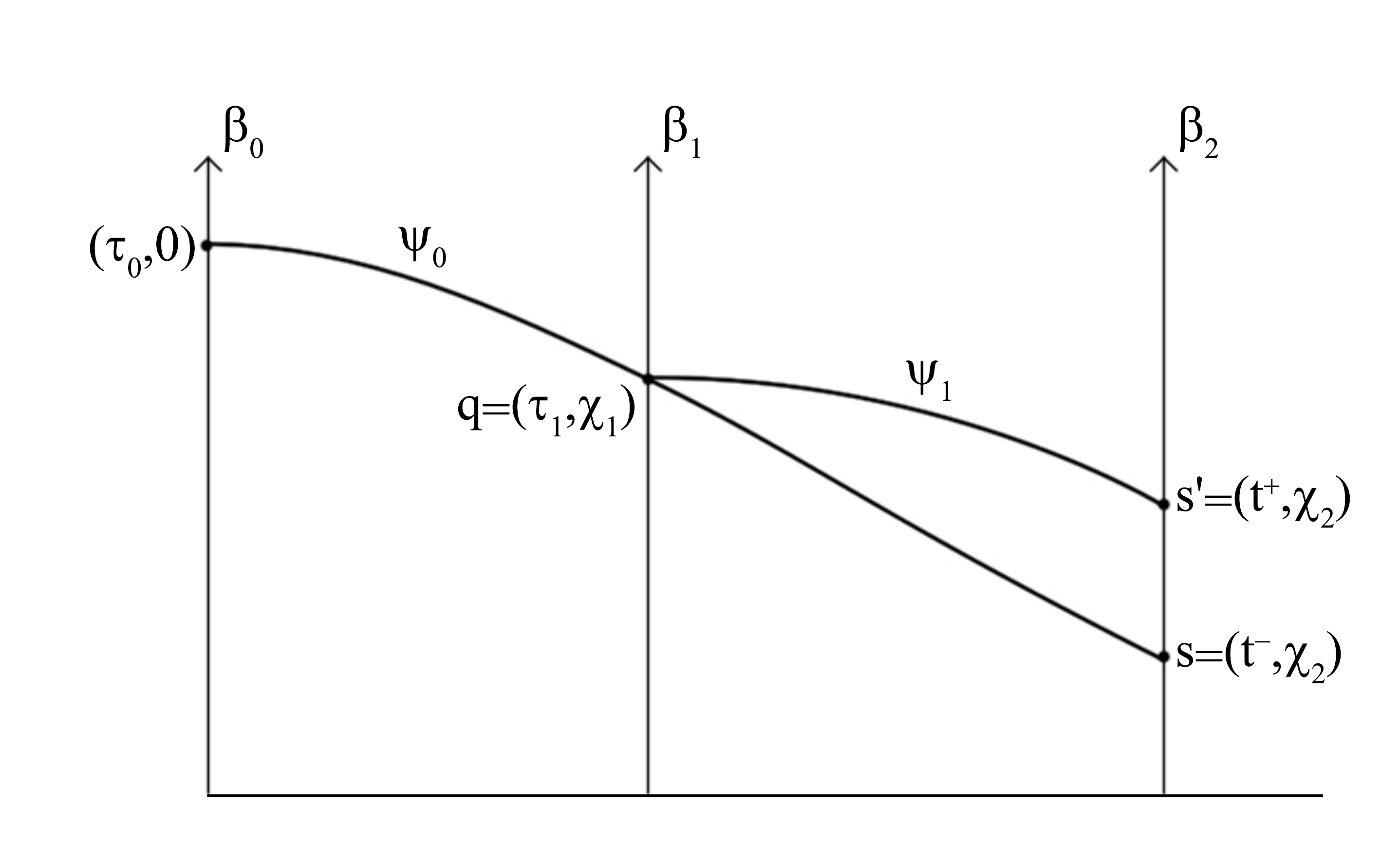}
\caption{Scenario I. Here $\Psi_0$ is the spacelike geodesic orthogonal to $\beta_0$, and $\Psi_1$ is the spacelike geodesic orthogonal to $\beta_1$. $(\tau_1,\chi_1)$ is the unique event $\Psi_0\cap\beta_1$ and $(\tau_2^+,\chi_2)$ is the unique event $\Psi_1\cap\beta_2$. The velocities of $\beta_2$ and $\beta_1$ relative to $\beta_0$ are taken at proper time $\tau_0$ of $\beta_0$, while the velocity of $\beta_2$ relative to $\beta_1$ is taken at proper time $\tau_1$ of $\beta_1$. Note that in this case the velocity of $\beta_2$ relative to the observers will be computed for separate events on the world line of $\beta_2$, i.e. the time discrepancy has been placed on the world line of the test particle.}\label{fc1}
\end{figure}

\begin{figure}[H]
\centering
\includegraphics[scale=0.3]{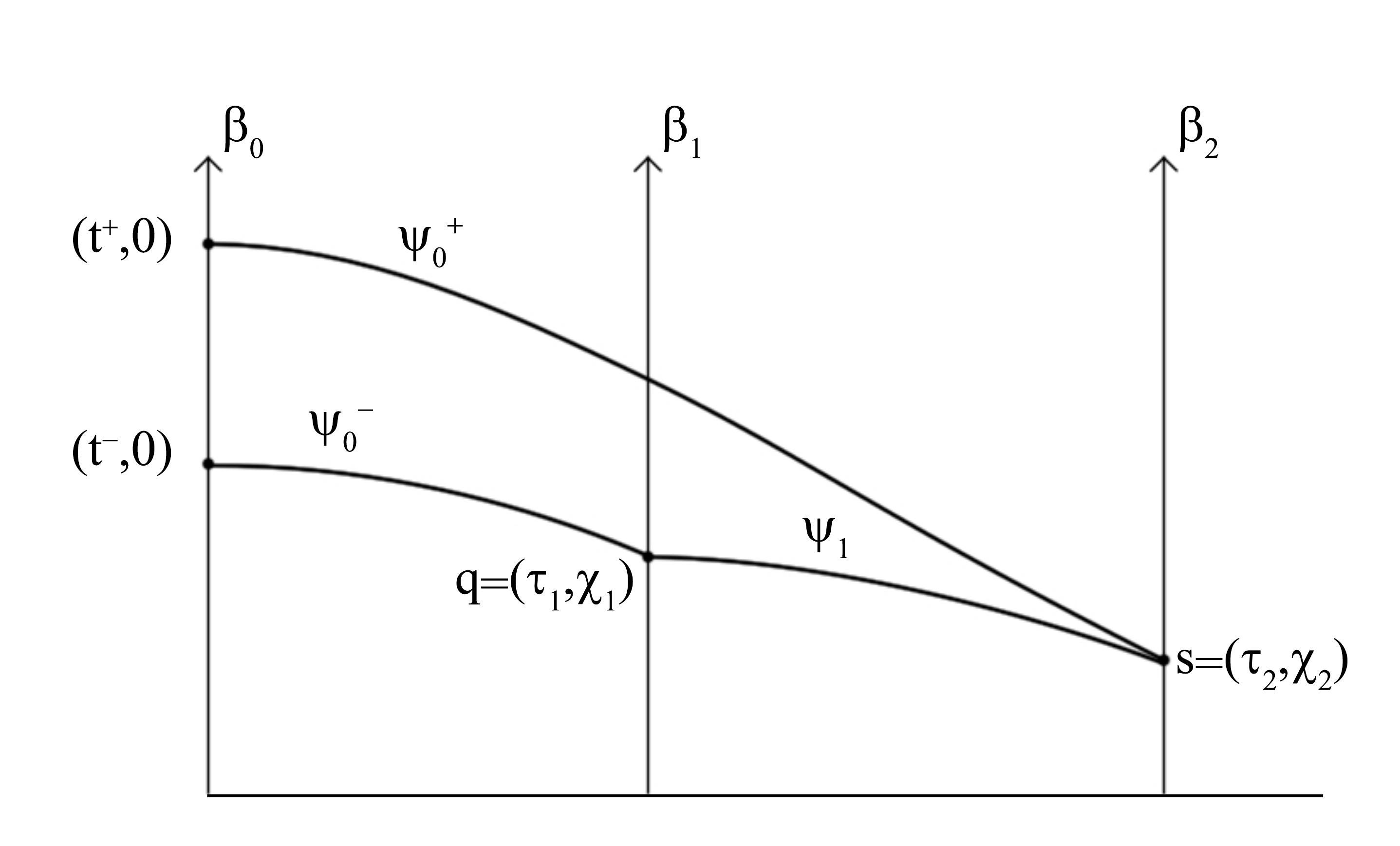}
\caption{Scenario II. Here $\Psi_0^+$ and $\Psi_0^-$ are the spacelike geodesics orthogonal to $\beta_0$ at $t=\tau_0^+$ and $t=\tau_0^-$, respectively. $(\tau_1,\chi_1)$ is the unique event $\Psi_0^-\cap\beta_1'$, and $\Psi_1$ is the spacelike geodesic orthogonal to $\beta_1$ emanating from this event. $\beta_1$, $\Psi_0^+$ and $\Psi_1$ intersect at the event $(\tau_2,\chi_2)$. Note that in this case the time discrepancy is on the world line of the central observer.}\label{fc2}
\end{figure}

\begin{figure}[H]
\centering
\includegraphics[scale=0.3]{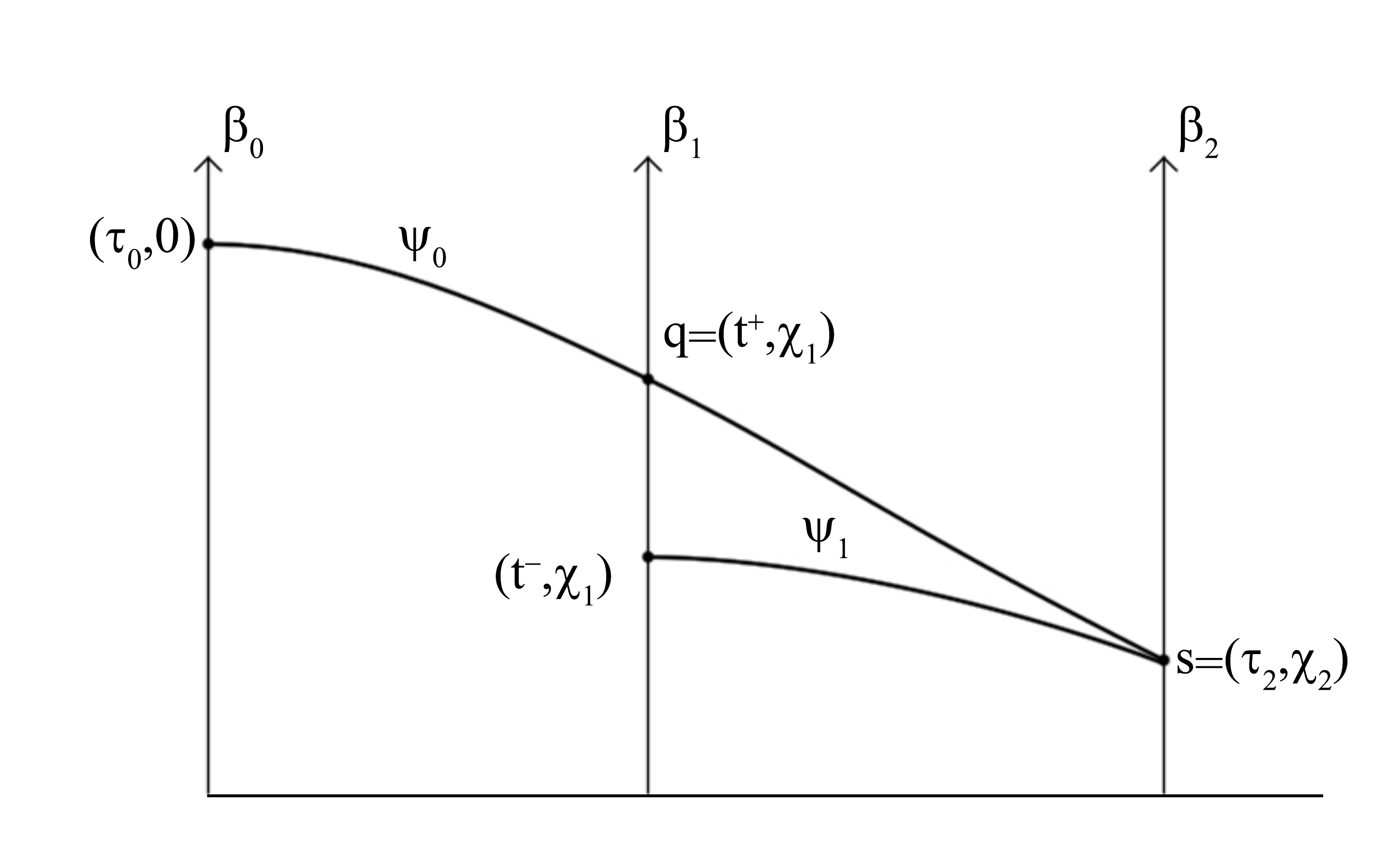}
\caption{Scenario III. Here $\Psi_0$ is a spacelike geodesic orthogonal to $\beta_0$, and intersects both $\beta_1$ and $\beta_2$ at $(\tau_1^+,\chi_1)$ and $(\tau_2,\chi_2)$, respectively. $\Psi_1$ is the unique geodesic orthogonal to $\beta_1$ such that $(\tau_2,\chi_2)=\beta_2\cap\Psi_1$. In this scenario the time discrepancy is on the world line of the secondary observer.}\label{fc3}
\end{figure}

\begin{remark}
Figures \ref{fc1}, \ref{fc2} and \ref{fc3} depict the case that $\chi_2>\chi_1>0$. However it is possible to connect the spacetime events shown in each diagram by spacelike geodesics for arbitrary configurations of the observers and test particle. The results we derive are thus valid in greater generality than the figures imply.
\end{remark}

\begin{theorem}\label{spacelike3}
For the scenarios depicted in Figures \ref{fc1}, \ref{fc2} and \ref{fc3}, let $v_{\mathrm{kin}1}$ be the kinematic speed of $\beta_{1}$ relative to $\beta_{0}$; $v_{\mathrm{kin}2}$ be the kinematic speed of $\beta_{2}$ relative to $\beta_{1}$; and $v_{\mathrm{kin}3}$ be the kinematic speed of $\beta_{2}$ relative to $\beta_{0}$.  Then,

\begin{equation}\label{veladdspace3}
(1-v_{\mathrm{kin}3}^{2})=\frac{a^{2}(t^{-})}{a^{2}(t^{+})}(1-v_{\mathrm{kin}1}^{2})(1-v_{\mathrm{kin}2}^{2}),
\end{equation}
where the proper times $t^{+}$ and $t^{-}$ are indicated in the figures and are determined by the indicated spacelike geodesics.
\end{theorem}

\begin{proof}
For fig. \ref{fc1}, using Eq. \eqref{eqvkin}, we may write,

\begin{equation}\label{3kin123}
v_{\mathrm{kin}1}^{2} =1-\frac{a^{2}(\tau_1)}{a^2(\tau_{0} )}\quad\quad
v_{\mathrm{kin}2}^{2} =1-\frac{a^{2}(t^+)}{a^2(\tau_1)}\quad\quad
v_{\mathrm{kin}3}^{2} =1-\frac{a^{2}(t^-)}{a^2(\tau_{0})}\\
\end{equation}
The result now follows by direct substitution of Eqs.\eqref{3kin123} into Eq.\eqref{veladdspace3}. Similarly for figures \ref{fc2} and \ref{fc3}, we may write

\begin{equation}
\label{2kin123}
v_{\mathrm{kin}1}^{2} =1-\frac{a^{2}(\tau_{1})}{a^2(t^- )}\quad\quad
v_{\mathrm{kin}2}^{2} =1-\frac{a^{2}(\tau_{2})}{a^2(\tau_{1})}\quad\quad
v_{\mathrm{kin}3}^{2} =1-\frac{a^{2}(\tau_{2})}{a^2(t^+)}\\
\end{equation}

\noindent and

\begin{equation}
\label{1kin123}
v_{\mathrm{kin}1}^{2} =1-\frac{a^{2}(t^+)}{a^2(\tau_{0})}\quad\quad
v_{\mathrm{kin}2}^{2} =1-\frac{a^{2}(\tau_{2})}{a^2(t^-)}\quad\quad
v_{\mathrm{kin}3}^{2} =1-\frac{a^{2}(\tau_{2})}{a^2(\tau_{0})},\\
\end{equation}
respectively. Substitution of Eq. \eqref{2kin123} or Eq. \eqref{1kin123} into Eq. \eqref{veladdspace3} verifies the identity.

\end{proof}

\noindent Using Eq.\eqref{vmetric}, the following corollary is immediate.

\begin{corollary}\label{Fermiadd}
Following the notational conventions of Sect.\ref{Fermicoords}, the Fermi relative velocity addition formula for comoving observers is given by,
\begin{equation}\label{veladdspace4}
\left(1+\frac{v_{\mathrm{Fermi}3}^{2}}{g_{\tau_{0}\tau_{0}}(s)}\right)=\frac{a^{2}(t^{-})}{a^{2}(t^{+})}\left(1+\frac{v_{\mathrm{Fermi}1}^{2}}{g_{\tau_{0}\tau_{0}}(q)}\right)\left(1+\frac{v_{\mathrm{Fermi}2}^{2}}{g_{\tau_{1}\tau_{1}}(s)}\right),
\end{equation}
where $s$ and $q$ stand for the appropriate Fermi coordinates for the spacetime points indicated in Figures \ref{fc1}, \ref{fc2} and \ref{fc3}, and where $g_{\tau_{1}\tau_{1}}(s)$ is replaced by $g_{\tau_{1}\tau_{1}}(s')$ for the scenario depicted in Figure \ref{fc1}.
\end{corollary}

\begin{example}\label{milne1}
The velocity addition formula for special relativity (more precisely for the Milne Universe) is a special case of Theorem \ref{spacelike3}.  We illustrate this for the configuration shown in Fig.\ref{fc3}. Similar calculations show that the result holds for the configurations shown in Figs. \ref{fc1} and \ref{fc2}. The transformation formulas from curvature coordinates $(t,\chi)$ to Fermi coordinates $(\tau,\rho)$ are,
\begin{equation}\label{transf}
\tau=t\cosh\chi \quad\quad \rho=t\sinh\chi
\end{equation}
(see e.g. \cite{Bolos12}).  Expressed in Fermi coordinates, the metric for the Milne Universe (Eq.\eqref{frwmetric} with $k=-1$ and $a(t)=t$) becomes the Minkowski metric, i.e., in two spacetime dimensions,
\begin{equation}
ds^{2}=-d\tau^{2}+d\rho^{2}.
\end{equation}
Applying Eq. \eqref{transf} to Fig.\ref{fc3} gives,
\begin{equation}\label{transf1}
\tau_{0}=t_{1}^{+}\cosh\chi_{1}=\tau_{2}\cosh\chi_{2}
\end{equation}
and
\begin{equation}\label{transf2}
t_{1}^{-}=\tau_{2}\cosh(\chi_{2}-\chi_{1}).
\end{equation}
Since $a(t)=t$, we then get,

\begin{equation}
\frac{a(t_{1}^{+})}{a(t_{1}^{-})}=\frac{\cosh(\chi_{2})}{\cosh(\chi_{2}-\chi_{1})\cosh\chi_{1}}=1+\tanh\chi_{1}\tanh(\chi_{2}-\chi_{1}).
\end{equation}
By combining Eqs. \eqref{transf1} and \eqref{transf2} with Eq.\eqref{eqvkin}, this may be expressed as,

\begin{equation}\label{afrac}
\frac{a^{2}(t_{1}^{+})}{a^{2}(t_{1}^{-})}= \left(1+v_{\mathrm{kin}1}v_{\mathrm{kin}2}\right)^{2}.
\end{equation}
Now substituting Eq. \eqref{afrac} into Eq. \eqref{veladdspace3} gives,

\begin{equation}
\left(1+v_{\mathrm{kin}1}v_{\mathrm{kin}2}\right)^{2}(1-v_{\mathrm{kin}3}^{2})=(1-v_{\mathrm{kin}1}^{2})(1-v_{\mathrm{kin}2}^{2}),
\end{equation}
which reduces to
\begin{equation}
v_{\mathrm{kin}3}=\frac{v_{\mathrm{kin}1}+v_{\mathrm{kin}2}}{1+v_{\mathrm{kin}1}v_{\mathrm{kin}2}}.
\end{equation}
Applying Corollary \ref{Fermiadd} with $g_{\tau\tau}\equiv-1$ for the Milne universe then recovers Eq.\eqref{special},
\begin{equation}\label{specialfermi}
v_{\mathrm{Fermi}3}=\frac{v_{\mathrm{Fermi}1}+v_{\mathrm{Fermi}2}}{1+v_{\mathrm{Fermi}1}v_{\mathrm{Fermi}2}}.
\end{equation}
We note that for a comoving observer with fixed coordinate $\chi$, it follows directly from Eq.\eqref{transf} that,
\begin{equation}
v_{\mathrm{Fermi}}=\tanh\chi.
\end{equation}
Thus, $\chi$ is the rapidity parameter for the Lorentz boost determined by $v_{\mathrm{Fermi}}$ (see, e.g., \cite{hobson}), and this Lorentz transformation is a hyperbolic rotation by angle $\chi$.  \noindent In \cite{Bolos12}, it was also shown that $v_{\mathrm{kin}}=v_{\mathrm{Fermi}}=\tanh\chi$ for a comoving test particle in the Milne universe, or, in the language of special relativity, a test particle with constant velocity departing from the central observer at time zero.

\end{example}

\section{Velocity addition formulas for lightlike simultaneity}\label{lightadd}

\noindent In this section we derive velocity addition formulas for the spectroscopic and astrometric relative velocities of comoving observers. For the observers $\beta_{0}, \beta_{1}, \beta_{2}$, the prototype configuration is depicted in Fig.\ref{fig1spec} in which the spacetime events $p, q$, and $s$ are all lightlike simultaneous. \\  

\begin{figure}[H]
\centering
\includegraphics[trim=1mm 10mm 1mm 20mm,clip,scale=0.3]{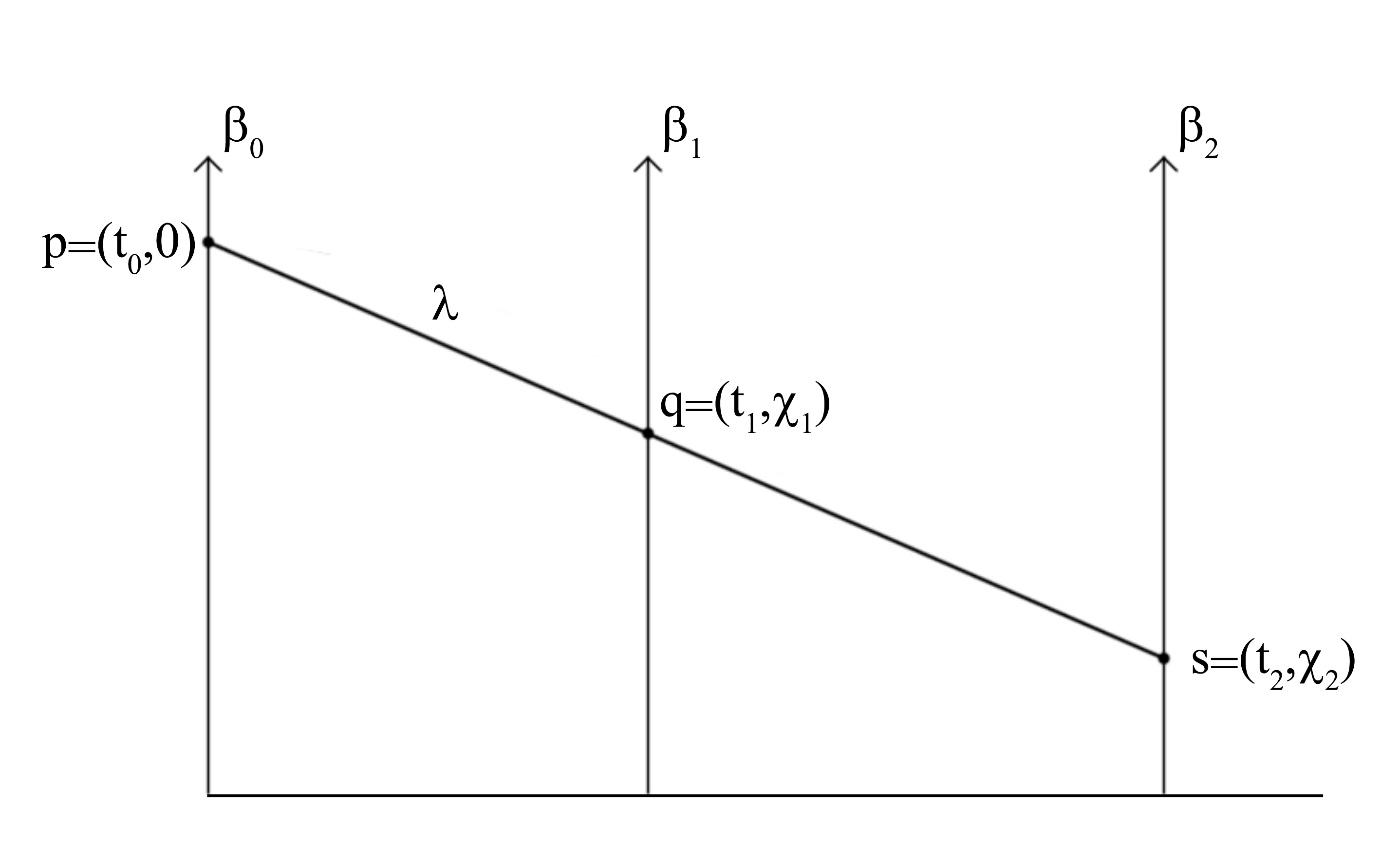}
\caption{Elements involved in the study of velocity addition for simultaneous (in optical coordinates) comoving observers. $\lambda$ is a lightlike geodesic.}\label{fig1spec}
\end{figure}

\begin{theorem}\label{spectheorem1}
For simultaneous (in optical coordinates) comoving observers the spectroscopic relative velocities are related by the special relativity addition formula, \eqref{special}, i.e.,
\begin{equation}\label{special2}
v_{\mathrm{spec}3}=\frac{v_{\mathrm{spec}1}+v_{\mathrm{spec}2}}{1+v_{\mathrm{spec}1}v_{\mathrm{spec}2}}
\end{equation}
\end{theorem}
\begin{proof}
For simultaneous comoving observers the spectroscopic velocities depicted in  Fig.\ref{fig1spec}, are given by \eqref{eqvspec} as,
\begin{equation}\label{vspec3}
v_{\mathrm{spec}1}=\frac{a^2(t_0)-a^2(t_1)}{a^2(t_0)+a^2(t_1)},\quad\quad v_{\mathrm{spec}2}=\frac{a^2(t_1)-a^2(t_2)}{a^2(t_1)+a^2(t_2)},
\end{equation}
and
\begin{equation}\label{vspec2}
v_{\mathrm{spec}3}=\frac{a^2(t_0)-a^2(t_2)}{a^2(t_0)+a^2(t_2)}.
\end{equation}
The result follows from direct substitution of Eqs. \eqref{vspec3} into the right side of Eq. \eqref{special2} which then becomes Eq. \eqref{vspec2}.
\end{proof}

\noindent Theorem \ref{spectheorem1} and Fig.\ref{fig1spec} assume that $\chi_2>\chi_1>0$. However, this restriction is not essential.  For any ordering of the spatial coordinates of the observers and test particle a measurement scheme can be found so that the spectroscopic velocities are related by Eq. \eqref{special}, but the measurements may not all be simultaneous. To illustrate, we prove this for the case that $\chi_1<0<\chi_2$. Fig. \ref{fig2spec} indicates how the spectroscopic relative velocities are to be measured for this case.

\begin{figure}[H]
\centering
\includegraphics[trim=1mm 10mm 1mm 20mm,clip,scale=0.3]{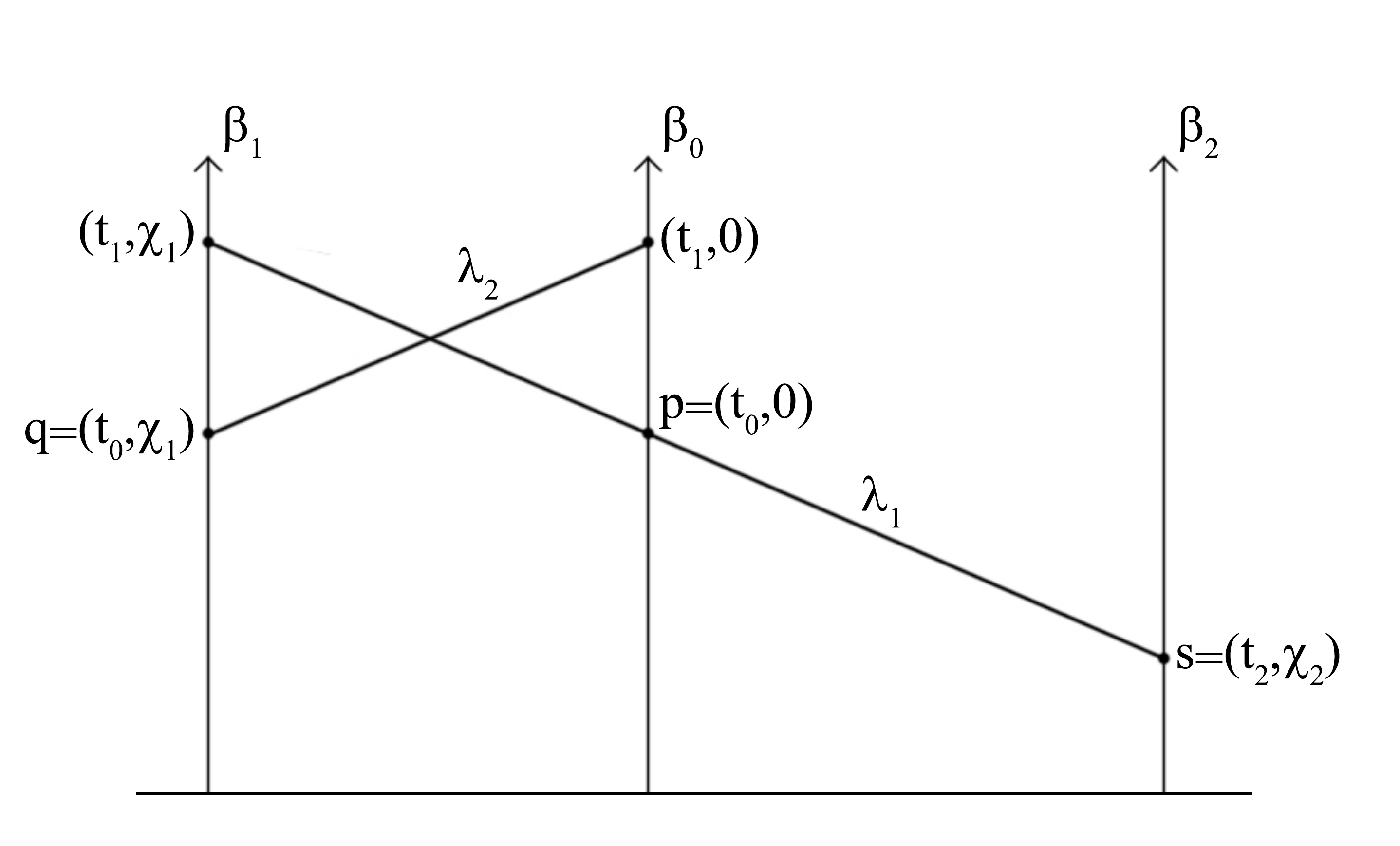}
\caption{In this scenario, $v_{\mathrm{spec}3}$ and $v_{\mathrm{spec}2}$ are the spectroscopic velocities of $\beta_2$ at the spacetime point $s$ relative to $\beta_0$ and $\beta_1$, respectively. $v_{\mathrm{spec}1}$ is the spectroscopic velocity of $\beta_1$ at the spacetime point $q$ relative to $\beta_0$. $\lambda_1$ and $\lambda_2$ are lightlike geodesics.}
\label{fig2spec}
\end{figure}

\begin{theorem}\label{spectheorem2}
The spectroscopic relative velocities of the comoving observers depicted in Figure \ref{fig2spec} are related by the special relativity addition formula, Eq. \eqref{special}, i.e.,
\begin{equation}\label{special2'}
v_{\mathrm{spec}3}=\frac{v_{\mathrm{spec}1}+v_{\mathrm{spec}2}}{1+v_{\mathrm{spec}1}v_{\mathrm{spec}2}}.
\end{equation}
\end{theorem}
\begin{proof}
For this scenario $v_{\mathrm{spec}2}$ and $v_{\mathrm{spec}3}$ are given in Eqs. \eqref{vspec3} and \eqref{vspec2}, respectively, while $v_{\mathrm{spec}1}$ is given by
\begin{equation}v_{\mathrm{spec}1}=-\left(\frac{a^2(t_1)-a^2(t_0)}{a^2(t_0)+a^2(t_1)}\right)=\frac{a^2(t_0)-a^2(t_1)}{a^2(t_0)+a^2(t_1)}.
\end{equation}
Thus the spectroscopic velocities are given by the same expressions as those used in the proof of Theorem \ref{spectheorem1}, and the result follows exactly as before.
\end{proof}

\noindent Using the velocity addition formula for the spectroscopic velocity and the functional relationship between the astrometric and spectroscopic velocities given by Eq. \eqref{eqvast}, we derive an addition formula for astrometric velocities for the lightlike simultaneity scenario depicted in Fig. \ref{fig1spec}.

\begin{lemma}\label{deltas}
Suppose that the spacetime points $\beta_{0}(t_{0})$, $\beta_{1}(t_{1})$, and $\beta_{2}(t_{2})$ lie on a null geodesic in the past pointing horismos $E_{t_{0}}^{-}$.  Denote the optical coordinates of $\beta_{1}(t_{1})$, and $\beta_{2}(t_{2})$ relative to the central observer $\beta_{0}$ by $(t_{0}, \delta_{1})$ and $(t_{0}, \delta_{2})$ respectively, and assume $\delta_{2}>\delta_{1}>0$ as in Fig. \ref{fig1spec}.  Let $\delta_{21}$ be the affine distance coordinate of $\beta_{2}(t_{2})$ in the optical coordinate system for $\beta_{1}$ (i.e. regarding $\beta_{1}$ as the central observer).  Then,

\begin{equation}
\delta_{2}=\delta_{1}+\frac{a(t_{1})}{a(t_{0})}\delta_{21}
\end{equation}
\end{lemma}

\begin{proof}
From Eq. \eqref{delta} with $\chi>0$,

\begin{align}
\delta_{2}&=\frac{1}{a(t_{0})}\int_{t_{2}}^{t_{0}}a(u)\mathrm{d}u\\
&=\frac{1}{a(t_{0})}\int_{t_{1}}^{t_{0}}a(u)\mathrm{d}u+\frac{a(t_{1})}{a(t_{0})}\frac{1}{a(t_{1})}\int_{t_{2}}^{t_{1}}a(u)\mathrm{d}u\\
&=\delta_{1}+\frac{a(t_{1})}{a(t_{0})}\delta_{21}.
\end{align}

\end{proof}

\begin{theorem}\label{asttheorem0}
Following the notation of Lemma \ref{deltas}, the astrometric relative velocities of the simultaneous comoving observers depicted in Figure \ref{fig1spec} are related by
\begin{equation}\label{astadd1}
-\left(2v_{\mathrm{ast}3}+\tilde{g}_{\textsc{t}\textsc{t}}(t_{0},\delta_{2})\right)=\left(2v_{\mathrm{ast}1}+\tilde{g}_{\textsc{t}\textsc{t}}(t_{0},\delta_{1})\right)\left(2v_{\mathrm{ast}2}+\tilde{g}_{\textsc{t}_{1}\textsc{t}_{1}}\left(t_{1}, \delta_{21}\right)\right),
\end{equation}
where $\tilde{g}_{\textsc{t}_{1}\textsc{t}_{1}}$ is the leading metric coefficient in optical coordinates for the secondary observer $\beta_{1}$ and
\begin{equation}
\delta_{21}=\frac{a(t_{0})}{a(t_{1})}(\delta_{2}-\delta_{1}).
\end{equation}
\end{theorem}

\begin{proof}
Taking into account Remark \ref{times} and combining Eqs. \eqref{vastend2} and \eqref{special2} (see also Eq. \eqref{special'}) gives,

\begin{equation}
\label{vastend3}
v_{\mathrm{ast}3}=-\frac{1}{2}\left(\tilde{g}_{\textsc{t}\textsc{t}}(t_{0},\delta_{2})+
\frac{1-v_{\mathrm{spec}1}}{1+v_{\mathrm{spec}1}}\cdot\frac{1-v_{\mathrm{spec}2}}{1+v_{\mathrm{spec}2}}\right). 
\end{equation}
It follows from Eq. \eqref{vastend2} that,
\begin{equation}\label{vastend4}
\frac{1-v_{\mathrm{spec}1}}{1+v_{\mathrm{spec}1}}=-2v_{\mathrm{ast}1}-\tilde{g}_{\textsc{t}\textsc{t}}(t_{0},\delta_{1})
\end{equation}
and
\begin{equation}\label{vastend5}
\frac{1-v_{\mathrm{spec}2}}{1+v_{\mathrm{spec}2}}=-2v_{\mathrm{ast}2}-\tilde{g}_{\textsc{t}_{1}\textsc{t}_{1}}\left(t_{1}, \delta_{21}\right).
\end{equation}
The result now follows by combining Eqs \eqref{vastend3},  \eqref{vastend4}, and  \eqref{vastend5}.
\end{proof}

\begin{example}  For the Milne Universe, it is easily verified from Eq. \eqref{opticalmetric} that $\tilde{g}_{\textsc{t}\textsc{t}}\equiv-1$.  Applying Eq. \eqref{astadd1} then gives the following special relativistic velocity addition formula for the astrometric relative velocity, for the configuration depicted in Fig.\ref{fig1spec},
\begin{equation}\label{milneastroadd}
v_{\text{ast}_3}=v_{\text{ast}_1}+v_{\text{ast}_2}-2v_{\text{ast}_1}v_{\text{ast}_2}.
\end{equation}
We note, however, that different velocity addition formulas for the astrometric relative velocity hold for configurations different from that depicted in Fig.\ref{fig1spec}.
\end{example}

\section{Hubble Additivity}\label{hubble}

\noindent In this section, we consider Robertson-Walker spacetimes for which the geometrically defined velocities of comoving test particles can be expressed as functions of the Hubble velocity.  Examples of such spacetimes are given in the following section.  In this situation, the addition formulas for the geometrically defined velocities may take simpler forms than for more general scenarios.\\ 

\noindent The usual foliation of a Robertson-Walker spacetime by maximally symmetric space slices $\{\Sigma_{t}\}$, parameterized by cosmological time $t$, determines a notion of simultaneity different from spacelike or lightlike simultaneity and naturally leads to the Hubble velocity and Hubble's law,

\begin{equation}\label{introhubble}
v\equiv\dot{d}(t)=\dot{a}(t)\chi=Hd.
\end{equation}
Here $v$ is Hubble velocity, $H$ is the Hubble parameter, $d=a(t)\chi$ is the proper distance on $\Sigma_{t}$ from the observer to the comoving test particle with coordinate $\chi$, and the overdot signifies differentiation with respect to $t$.\\ 

 \noindent The velocity addition formula for the Hubble velocity of comoving objects, relative to a comoving observer is particularly simple.  Following the notation of the previous sections, let the secondary oberver, $\beta_1$, whose fixed spatial coordinate is $\chi_1$, have Hubble speed $v_{1}$ relative to $\beta_{0}$. The Hubble speed of $\beta_2$, at $\chi_2$, relative to $\beta_{0}$ is denoted by $v_{3}$, and $v_{2}$ denotes the Hubble speed of $\beta_2$ relative to $\beta_1$.  Then at cosmological time $\tau$,
 
 \begin{equation}\label{hubbleadd}
v_3 = \dot{a}(\tau)\chi_2 = \dot{a}(\tau)[\chi_1 +(\chi_2-\chi_1)]= v_1 + v_2.
\end{equation}

\noindent We see that the velocity addition formula for the Hubble velocity is Galilean.\\ 

\noindent In general (see \cite{Bolos12}), all four of the geometrically defined relative velocities  of a comoving test particle are uniquely determined by the two coordinates, $(\tau, \chi)$, where $\chi$ is the fixed coordinate of the comoving test particle, and $\tau$ is the proper time of the central observer.  However, in some cases, the dependence of the relative velocities is exclusively through the Hubble speed parameter $v$,
\begin{equation}
\label{hubblespeed}
v=v(\tau ,\chi )\equiv\dot{a}(\tau )\chi.
\end{equation} 
\begin{remark}
\label{remhubblespeed}
The expression \eqref{hubblespeed} for $v$ is the Hubble speed of a comoving test particle with curvature-normalized coordinates $(\tau ,\chi )$.  However, it is important to recognize that the relative velocities that we calculate in this section, as functions of $v$, are those of test particles located at different spacetime points, of the form $(t,\chi)$, where $t<\tau$. 
\end{remark}

\noindent If there is a one-to-one correspondence between Hubble velocities to geometric velocities (Fermi, kinematic, spectroscopic, or astrometric), then an addition formula for the geometrically defined relative velocities may be derived from Eq.\eqref{hubbleadd} as the following theorem shows.

\begin{theorem}\label{hubaddth}
Let $v_s$ denote one of the four geometrically defined relative velocities considered in this paper, and let $v$ denote the Hubble velocity of a test particle. Suppose further that $v_s$ has a 1-1 functional dependence on $v$, $v_s=f(v)$. Then the geometric velocity, $v_{s3}$, of $\beta_2$ relative to $\beta_0$ at proper time $\tau_0^+$ of $\beta_0$ can be written in terms of $v_{s1}$ and $v_{2s}$, the geometric relative velocities of $\beta_1$ relative to $\beta_0$ and $\beta_2$ relative to $\beta_1$, respectively. Here $v_{s1}$ is the relative velocity of $\beta_1$ at proper time $\tau_0^-$ (it is possible for $\tau_0^-=\tau_0^+$) of $\beta_0$ and $v_{s2}$ is the relative velocity of $\beta_2$ at proper time $\tau_1$ of $\beta_1$. The relationship is
\begin{equation}\label{hubadd}
v_{s3}=f\left(\frac{\dot{a}(\tau_0^+)}{\dot{a}(\tau_0^-)}f^{-1}(v_{s1})+\frac{\dot{a}(\tau_0^+)}{\dot{a}(\tau_1)}f^{-1}(v_{s2})\right)
\end{equation}
\end{theorem}

\begin{proof}
The Hubble velocity, $v_3$, of $\beta_2$ relative to $\beta$ at $\tau_0^+$ can be expressed as,
\begin{align}\label{hubaddp1}
v_3=\dot{a}(\tau_0^+)\chi_2&=\dot{a}(\tau_0^+)\chi_1+\dot{a}(\tau_0^+)(\chi_2-\chi_1)\\
&=\frac{\dot{a}(\tau_0^+)}{\dot{a}(\tau_0^-)}\dot{a}(\tau_0^-)\chi_1+\frac{\dot{a}(\tau_0^+)}{\dot{a}(\tau_1)}\dot{a}(\tau_1)(\chi_2-\chi_1)\nonumber \\
&=\frac{\dot{a}(\tau_0^+)}{\dot{a}(\tau_0^-)}v_1+\frac{\dot{a}(\tau_0^+)}{\dot{a}(\tau_1)}v_2,\nonumber
\end{align}
where $v_1$ is the Hubble velocity of $\beta_1$ relative to $\beta_0$ at proper time $\tau_0^-$ of $\beta_0$ and $v_2$ is the Hubble velocity of $\beta_2$ relative to $\beta_1$ at proper time $\tau_1$ of $\beta_1$. By assumption $v_{s1}=f(v_1)$ and $v_{s2}=f(v_2)$. Since $f$ is 1-1, we may write $v_1=f^{-1}(v_{s1})$ and $v_2=f^{-1}(v_{s2})$. Substituting these expressions into Eq. \eqref{hubaddp1} we obtain,

\begin{equation}\label{hubaddp2}
v_3=\frac{\dot{a}(\tau_0^+)}{\dot{a}(\tau_0^-)}f^{-1}(v_{s1})+\frac{\dot{a}(\tau_0^+)}{\dot{a}(\tau_1)}f^{-1}(v_{s2}).
\end{equation}

\noindent Then from Eq. \eqref{hubaddp2}, we have that,
\begin{equation}
v_{s3}=f(v_3)=f\left(\frac{\dot{a}(\tau_0^+)}{\dot{a}(\tau_0^-)}f^{-1}(v_{s1})+\frac{\dot{a}(\tau_0^+)}{\dot{a}(\tau_1)}f^{-1}(v_{s2})\right).
\end{equation}
\end{proof}

\noindent While Eq. \eqref{hubadd} is less general than the other formulas developed so far, it is often easier to use when working with a particular spacetime where it is applicable.\\

\noindent We conclude this section with a corollary that gives a kinematic velocity addition formula for Robertson-Walker spacetimes with scale factors following power laws, i.e., scale factors of the form,

\begin{equation}
a(t)=t^{\alpha},
\end{equation}
where $\alpha >0$, $\alpha\neq 1$\footnote{The case $\alpha=1$ gives the Milne universe}. It was shown in \cite{randles} and \cite{sam}  that the proper radius $\rho_{\mathcal{M}_{\tau}}(\alpha)$ of $\mathcal{M}_{\tau}$ is an increasing function of $\tau$ and,
\begin{equation}
\label{calpha2}
C_{\alpha }\equiv\frac{\rho_{\mathcal{M}_{\tau}}(\alpha)}{\tau} =\frac{\sqrt{\pi }\,\Gamma \left( \frac{1+\alpha }{2\alpha }\right) }{\Gamma \left( \frac{1}{2\alpha }\right)}.
\end{equation}
 Let $F_{\alpha }(z) \equiv z^{\frac{1-\alpha }{\alpha }} \, _2F_1\left( \frac{1}{2},\frac{1-\alpha }{2\alpha };\frac{1+\alpha }{2\alpha };z^2\right)$
where $0<z<1$, and where $ _2F_1(\cdot ,\cdot ;\cdot ;\cdot )$ is the Gauss hypergeometric function.  Define $G_{\alpha }(v)\equiv\left( F_{\alpha }^{-1}\left( C_{\alpha }+\frac{\alpha -1}{\alpha}v\right) \right)^{1/\alpha }$.  Specializing to the case $\chi>0$ (for convenience only), for a comoving test particle (see \cite{sam}) then,
\begin{equation}
\label{powerVkinetic}
 v_{\mathrm{kin}} = f(v)= \sqrt{1-G_{\alpha }^{2\alpha }(v)},
\end{equation}
where the function $f$ plays the same role here as in Theorem \ref{hubaddth}.  For $0<\alpha<1$, $v$ is bounded above by $\frac{\alpha}{1-\alpha}C_{\alpha }$, but for $\alpha>1$, $\chi $ and $v$ have no upper bounds \cite{sam}. It is readily seen that,

\begin{equation}\label{-1}
v=f^{-1}(v_{\mathrm{kin}})=\frac{\alpha}{\alpha-1}\left[F_{\alpha}\left(\sqrt{1-v_{\text{kin}}^2}\right)-C_{\alpha} \right].
\end{equation}

\begin{corollary}\label{powerkinadd}
Let the scale factor for Robertson-Walker spacetime be given by $a(t)=t^{\alpha}$ with  $\alpha >0$, $\alpha\neq 1$. For the scenario depicted in Figure \ref{fc1}, with $\chi_2>\chi_1>0$, the kinematic velocity addition formula is given by,
\begin{equation}\label{poweradd}
v_{\text{kin}3}= \sqrt{1-G^{2\alpha}_{\alpha}\left(f^{-1}(v_{\text{kin}1})+(1-v_{\text{kin}1}^2)^{\frac{\alpha-1}{2\alpha}}f^{-1}(v_{\text{kin}2})\right)}.
\end{equation}
wherein we refer to Eq. \eqref{-1}.
\end{corollary}
\begin{proof}
From \eqref{eqvkin} we obtain,
\begin{equation}
\frac{\tau_1^{\alpha}}{\tau_{0}^{\alpha}}=\frac{a(\tau_1)}{a(\tau_{0})}=\sqrt{1-v_{\text{kin}1}^2},
\end{equation}
and thus,
\begin{equation}\label{dot}
\frac{\dot{a}(\tau_1)}{\dot{a}(\tau_{0})}=\frac{\tau_1^{\alpha-1}}{\tau_{0}^{\alpha-1}}=\left(1-v_{\text{kin}1}^{2}\right)^{\frac{\alpha-1}{2\alpha}}.
\end{equation}
The specialization of Eq. \eqref{hubadd} to the kinematic relative velocity in the configuration depicted by Figure \ref{fc1} is, 

\begin{equation}
v_{\text{kin}3}=f\left(f^{-1}(v_{\text{kin}1})+\frac{\dot{a}(\tau_0)}{\dot{a}(\tau_1)}f^{-1}(v_{\text{kin}2})\right)
\end{equation} 
Combining this with Eqs.\eqref{powerVkinetic} and \eqref{dot} then yields Eq. \eqref{poweradd}. 
\end{proof}
\noindent We mention that the cases of negative $\chi$  values, and/or $\chi_1>\chi_2$ are also easily handled at the expense of some absolute values, and that similar results for the cases of Figures \ref{fc2} and \ref{fc3} may be obtained using the same method.

\begin{remark}
The velocity addition formula given by Corollary \ref{powerkinadd} is more specialized than the formula given in Theorem \ref{spacelike3}, but, as in special relativity, has the feature that there is no explicit time coordinate dependence, unlike the more general velocity addition formula given in Theorem \ref{spacelike3}.
\end{remark}

\section{Examples}\label{examples}

In this section we apply the results of the previous sections to find explicit expressions for the velocity addition formulas of comoving observers in particular Robertson-Walker spacetimes.

\subsection{The de Sitter universe}
The Hubble velocity for a comoving test particle in the de Sitter universe is given by,

\begin{equation}
\label{eqvHdsu}
v=v(\tau ,\chi ) =H_{0} e^{H_{0} \tau }\chi.
\end{equation}
where $H_{0}$ is the Hubble constant. Expressed in the notation of Theorem \ref{hubaddth}, it was shown in \cite{Bolos12} that the kinematic velocity of a test particle relative to a central observer is given by 
\begin{equation}\label{fdesitter}
v_{\mathrm{kin}}=f(v)=\frac{v}{\sqrt{1+v^2}},
\end{equation}
for $|v(\tau ,\cdot )|<\sqrt{e^{2H_0 \tau }-1}$ and $|v(\cdot ,\chi )|>\frac{H_0\chi }{\sqrt{1-\left( H_0 \chi \right) ^2}}$.  Then,
\begin{equation}\label{dsvkinv}
v=f^{-1}(v_{\mathrm{kin}})=\frac{v_{\mathrm{kin}}}{\sqrt{1-v_{\mathrm{kin}}^2}}.
\end{equation}
\noindent We use Theorem \ref{hubaddth} to find a velocity addition formula for the kinematic relative velocity in the scenarios depicted in Figs. \ref{fc1} and \ref{fc2}.  Applying Eq. \eqref{hubaddp1} to the scenario depicted in Fig. \ref{fc1} we have that
\begin{equation}\label{dsvk1}
v_3=H_0e^{H_0\tau_0}\chi_2=v_1+\frac{\dot{a}(\tau_0)}{\dot{a}(\tau_1)}v_2.
\end{equation} 

\noindent Using \eqref{eqvkin} we have,

\begin{equation}\label{dsvk2}
\frac{\dot{a}(\tau_0)}{\dot{a}(\tau_1)}=\frac{e^{H_0\tau_0}}{e^{H_0\tau_1}}=\frac{a(\tau_0)}{a(\tau_1)}=\frac{1}{\sqrt{1-v_{\text{kin}1}^2}}.
\end{equation}

\noindent Combining Eqs. \eqref{dsvk1} and \eqref{dsvk2} and substituting them into the formula $v_{\text{kin}3}=\frac{v_3}{\sqrt{1+v_3^2}}$ along with Eq. \eqref{dsvkinv} yields the addition formula,

\begin{equation}
v_{\text{kin}3}=\frac{v_{\text{kin}1}\sqrt{1-v_{\text{kin}2}^2}+v_{\text{kin}2}}{\sqrt{1+2v_{\text{kin}1}v_{\text{kin}2}\sqrt{1-v_{\text{kin}2}^2}}}.
\end{equation}

\noindent Turning to the scenario depicted in Fig. \ref{fc2}, we have by Eq. \eqref{hubaddp1} that

\begin{equation}\label{dsvk3}
v_3=\frac{\dot{a}(t^+)}{\dot{a}(t^-)}v_1+\frac{\dot{a}(t^+)}{\dot{a}(\tau_1)}v_2.
\end{equation}

\noindent Again, using \eqref{eqvkin} we have,

\begin{align}
\label{dsvk4}\frac{\dot{a}(\tau_2)}{\dot{a}(t^{+})}&=\frac{e^{H_0\tau_2}}{e^{H_0t^{+}}}=\frac{a(\tau_2)}{a(t^{+})}=\sqrt{1-v_{\text{kin}3}^2}\\
\label{dsvk5}\frac{\dot{a}(\tau_2)}{\dot{a}(\tau_{1})}&=\frac{e^{H_0\tau_2}}{e^{H_0\tau_{1}}}=\frac{a(\tau_2)}{a(\tau_{1})}=\sqrt{1-v_{\text{kin}2}^2}\\
\label{dsvk6}\frac{\dot{a}(\tau_1)}{\dot{a}(t^{-})}&=\frac{e^{H_0\tau_1}}{e^{H_0t^{-}}}=\frac{a(\tau_1)}{a(t^{-})}=\sqrt{1-v_{\text{kin}1}^2}.
\end{align}
Substituting Eqs. \eqref{dsvk4}, \eqref{dsvk5}, \eqref{dsvk6} into Eq. \eqref{dsvk3} and then using Eq. \eqref{dsvkinv} gives the velocity addition formula,

\begin{equation}
v_{\mathrm{kin}3}=v_{\mathrm{kin}1}\sqrt{1-v_{\mathrm{kin}2}^2}+v_{\mathrm{kin}2}.
\end{equation}

\subsection{Addition Formulas for $a(t)=t^{1/3}$}

For this scale factor,

\begin{equation}
v=\frac{\chi}{3\tau^{2/3}},
\end{equation}
and it was shown in \cite{Bolos12} that, 

\begin{equation}\label{a13kin}
v_{\mathrm{kin}}=v
\end{equation}

\noindent and 

\begin{equation}\label{a13fermi}
v_{\mathrm{Fermi}}=v(1+v^2),
\end{equation}

\noindent with $|v|<1$. Since the function defined by Eq. \eqref{a13fermi} is 1-1, the Hubble velocity can be expressed in terms of the Fermi relative velocity as
\begin{equation}\label{a13vfinv}
v(v_\mathrm{Fermi})=\frac{1}{C(v_\mathrm{Fermi})}-\frac{1}{3}C(v_\mathrm{Fermi}),
\end{equation}

\noindent where

\begin{equation}
C(v_\mathrm{Fermi})=\left(\frac{-27v_\mathrm{Fermi}+\sqrt{27(4+27v_{\mathrm{Fermi}}^2)}}{2}\right)^{1/3}.
\end{equation}

\noindent We apply Theorem \ref{hubaddth} to the scenario depicted in Fig. \ref{fc1} to obtain addition formulas for the Fermi and kinematic relative velocities. From Eqs. \eqref{eqvkin},
\begin{equation}
\left(\frac{\tau_1}{\tau_{0}}\right)^{1/3}=\frac{a(\tau_1)}{a(\tau_{0})}=\sqrt{1-v_{\text{kin}1}^2}.
\end{equation}
Combining this with Eq. \eqref{a13kin} gives,
\begin{equation}
\tau_1=(1-v_1^2)^{3/2}\tau_0.
\end{equation}

\noindent Eq. \eqref{hubaddp1} for this situation therefore reads,

\begin{equation}\label{a13d}
v_3=v_1+\frac{\dot{a}(\tau_0)}{\dot{a}(\tau_1)}v_2=v_1+(1-v_1^2)v_2.
\end{equation}

\noindent For the Fermi relative velocity it follows from Eq. \eqref{a13d} that

\begin{equation}
v_{\mathrm{Fermi}_3}=\big[v(v_{\mathrm{F}1})+(1-v(v_{\mathrm{F}1})^2)v(v_{\mathrm{F}2})\big] \big[1+\big(v(v_{\mathrm{F}1})+(1-v(v_{\mathrm{F}1})^2)v(v_{\mathrm{F}2})\big)^2\big],
\end{equation}

\noindent where $v(v_{\mathrm{F}j})\equiv v(v_{\text{Fermi}j})$ is given by Eq. \eqref{a13vfinv} for $j=1,2$.\newline

\noindent For the kinematic relative velocity, we obtain the velocity addition formula by use of Eq. \eqref{a13kin} in Eq. \eqref{a13d}. The result is, 

\begin{equation}
v_{\text{kin}3}=v_{\text{kin}1}+(1-v_{\text{kin}1}^2)v_{\text{kin}2}.
\end{equation}

\subsection{Addition Formulas for $a(t)=t^{1/2}$}
For the scale factor $a(t)=t^{1/2}$ for the radiation dominated universe, the Hubble speed is given by,

\begin{equation}
v=\frac{\chi}{2\sqrt{\tau}},
\end{equation}
and it was shown in $\cite{Bolos12}$, that the kinematic and Fermi relative velocities of a comoving test particle are given by,

\begin{equation}\label{sinv}
v_\text{kin}=\sin v
\end{equation}
and

\begin{equation}\label{a12vf}
v_\text{Fermi}=(\cos v+v\sin v)\sin v,
\end{equation}

\noindent where $|v|<\pi/2$. From Eqs. \eqref{eqvkin}, for the scenario depicted in Fig. \ref{fc1},
\begin{equation}
\left(\frac{\tau_1}{\tau_{0}}\right)^{1/2}=\frac{a(\tau_1)}{a(\tau_{0})}=\sqrt{1-v_{\text{kin}1}^2}.
\end{equation}
Combining this with Eq. \eqref{sinv} gives,

\begin{equation}
\tau_1=\tau_0\cos^2v_1.
\end{equation}
Therefore,

\begin{equation}
v_3=v_1+\frac{\dot{a}(\tau_0)}{\dot{a}(\tau_1)}v_2=v_1+v_2\cos v_1
\end{equation}

\noindent For the kinematic relative velocity addition formula, we then have that

\begin{align}
v_{\text{kin}_3}&=\sin v_3=\sin(v_1+v_2\cos v_1)\nonumber\\
&=\sin(v_1)\cos(v_2\cos v_1)+\cos(v_1)\sin(v_2\cos v_1)\nonumber\\
&=v_{\text{kin}_1}\cos\left(\sin^{-1}(v_{\text{kin}_2})\sqrt{1-v_{\text{kin}_1}^2}\right)\\
&+\sqrt{1-v_{\text{kin}_1}^2}\sin\left(\sin^{-1}(v_{\text{kin}_2})\sqrt{1-v_{\text{kin}_1}^2}\right).\nonumber
\end{align}

\noindent For the Fermi relative velocity, since the function defined by Eq. \eqref{a12vf} is 1-1 on the interval $(-\pi/2,\pi/2)$, in principle an inverse function can be found that gives the Hubble velocity as a function of the Fermi relative velocity. By Theorem \ref{hubaddth} we can then find an addition formula for the scale factor $a(t)=t^{1/2}$.\\

\noindent Using the results of $\cite{Bolos12, sam}$, general velocity addition formulas can be similarly formulated for scale factors of the form, $a(t)=t^{\alpha}$ for $\alpha >0$.

\section{Conclusion}\label{conclude}

Natural generalizations of the the special relativistic velocity addition formula arise from the consideration of comoving observers and particles in Robertson-Walker spacetimes.  As described in the introduction, such a configuration is the natural generalization of special relativistic inertial frames in standard configuration, related by Lorentz boosts.\\

\noindent The introduction of geometrically defined relative velocities, Fermi, kinematic, spectroscopic, and astrometric came about following discussions about the need for a strict definition of \textquotedblleft radial
velocity\textquotedblright\ at the General Assembly of
the International Astronomical Union (IAU), held in 2000 (see, e.g.,
\cite{Bolos02,Soff03,Lind03}).  Given the observation of the preceding paragraph, it is natural to investigate addition formulas for these geometric relative velocities.\\

\noindent General velocity addition formulas for the Fermi and kinematic relative velocities, associated with spacelike simultaneity were found in Section \ref{spacelikesection} and the analogs for the spectroscopic and astrometric relative velocities associated with lightlike simultaneity were developed in Section \ref{lightadd}.  In this degree of generality, with the important exception of the spectroscopic relative velocity (see Theorems \ref{spectheorem1} and \ref{spectheorem2}), the velocity addition formulas depend, not only on the space coordinate $\chi$ (in curvature coordinates) of the comoving observers and particles, but also on their time coordinates.\\

\noindent However, if a geometrically defined relative velocities is a bijective function of the Hubble velocity, in the sense described in Section \ref{hubble}, then the dependence on time coordinates in the addition formulas can, at least in some cases, be eliminated.  This is carried out for the de Sitter universe and general power law cosmologies with scale factors of the form $a(t)=t^{\alpha}, \alpha>0$ ($\alpha = 1$ gives the Milne universe, i.e., special relativity for which the usual velocity addition formula is recovered by our methods). Simple examples, including the radiation dominated universe, are given in Section \ref{examples}.\\ 

\noindent When $\alpha>1$, the Robertson-Walker cosmologies with scale factor $a(t)=t^{\alpha}$ include event horizons.  Fermi and optical coordinates cannot be extended beyond the event horizon \cite{klein13}, so the geometric relative velocities of test particles in that region of spacetime are undefined. By contrast, Hubble velocity fields beyond the event horizon have been studied \cite{johnston} and are of interest as part of the large scale structure of the universe.  A possible future direction for research could be to investigate whether the velocity addition laws in this paper could be used to extend the definitions of geometric relative velocities through the use of an intermediary observer, within the central observers event horizon, in a way that is well defined, and for more realistic spacetimes.  \\ 

\noindent \textbf{Acknowledgment.} J. Reschke was partially supported during the course of this research by the Interdisciplinary Research Institute for the Sciences at California State University, Northridge.

\end{document}